\documentclass[letterpaper, conference]{IEEEtran}  
\def\IEEEtitletopspace{18pt} 
\IEEEoverridecommandlockouts                              

\usepackage{cite}
\usepackage{amsmath,amssymb,amsfonts,amsthm}
\usepackage[italicdiff]{physics}
\usepackage{mathtools}
\usepackage{bm}
\usepackage{bbm}
\usepackage{mathrsfs}
\usepackage{booktabs}
\usepackage{xspace}

\usepackage{tikz}
\usepackage{pgfplots}
\pgfplotsset{compat=1.18}
\usetikzlibrary{positioning, angles, quotes, math, shapes, arrows.meta, arrows, patterns, decorations.pathmorphing, decorations.markings}

\newcommand{\card}[1]{\textrm{cardinality}(#1)}
\newcommand{\conv}[1]{\textrm{conv}(#1)}
\newcommand{\MVEE}[1]{\textrm{MVEE}(#1)}
\renewcommand{\det}[1]{\textrm{det}(#1)}

\newcommand{\fh}[0]{\hat{f}}
\newcommand{\gh}[0]{\hat{g}}
\newcommand{\xh}[0]{\hat{x}}
\newcommand{\Xh}[0]{\hat{X}}
\newcommand{\yh}[0]{\hat{y}}

\newcommand{\wh}[0]{\hat{w}}
\newcommand{\Wh}[0]{\hat{W}}
\newcommand{\vh}[0]{\hat{v}}

\newcommand{\mathcalWh}[0]{\hat{\mathcal{W}}}

\newcommand{\mathcalVh}[0]{\hat{\mathcal{V}}}

\newcommand{\thetah}[0]{\hat{\theta}}
\newcommand{\Thetah}[0]{\hat{\Theta}}
\newcommand{\gammab}[0]{\bar{\gamma}}
\newcommand{\gammah}[0]{\hat{\gamma}}

\newcommand{\Gammah}[0]{\hat{\Gamma}}

\newcommand{\xb}[0]{\bar{x}}
\newcommand{\Xb}[0]{\bar{X}}

\newcommand{\wb}[0]{{\bar{w}}}
\newcommand{\Wb}[0]{{\bar{W}}}

\newcommand{\thetab}[0]{{\bar{\theta}}}

\newcommand{\R}[1]{\mathbb{R}^{#1}}
\newcommand{\justR}[0]{\mathbb{R}}

\DeclareMathOperator*{\argmin}{argmin}

\newcommand{\ie}[0]{i.\,e.\xspace}

\newcommand{\st}[0]{s.\,t.\xspace}

\newtheorem{theorem}{Theorem}[section]

\newtheorem{lemma}[theorem]{Lemma}

\newtheorem{proposition}{Proposition}[section]

\theoremstyle{definition}
\newtheorem{definition}{Definition}[section]
\newtheorem{fact}{Fact}[section]

\theoremstyle{plain}

\theoremstyle{remark}
\newtheorem*{remark}{Remark}

\begin{document}

\title{\LARGE\bf Optimistic vs Pessimistic Uncertainty Model Unfalsification}

\author{Jannes H\"uhnerbein$^{1}$, Jad Wehbeh$^{2}$ and Eric C. Kerrigan$^{3}$
\thanks{$^{1}$Jannes H\"uhnerbein is with the Chair of Astrodynamics, Technical University of Munich, 85521 Ottobrunn, Germany
        {\tt\small jannes.huehnerbein@tum.de}}%
\thanks{$^{2}$Jad Wehbeh is with the Department of Electrical and Electronic Engineering, Imperial College London, SW7 2AZ, UK
        {\tt\small j.wehbeh22@imperial.ac.uk}}%
\thanks{$^{3}$Eric C. Kerrigan is with the Department of Electrical and Electronic Engineering and the Department of Aeronautics, Imperial College London, SW7 2AZ, UK
        {\tt\small e.kerrigan@imperial.ac.uk}}%
}

\maketitle

\begin{abstract}
We present a novel, input-output data-driven approach to uncertainty model identification.
As the true bounds and distributions of system uncertainties ultimately remain unknown, we depart from the goal of identifying the uncertainty model and instead
look for minimal concrete statements that can be made based on an uncertain system
model and available input-output data. We refer to this as unfalsifying an uncertainty model.
Two different unfalsification approaches are taken.
The optimistic approach determines the smallest uncertainties that could
explain the given data, while the pessimistic approach finds the largest possible uncertainties suggested by the data.
The pessimistic problem is revealed to be a semi-infinite program, which is solved using the local reduction algorithm. It is also shown that the optimistic and pessimistic approaches to uncertainty model unfalsification are mathematical duals. Finally, both approaches are tested using an uncertain linear model with data from a simulated nonlinear system.
\end{abstract}

\section{Introduction}
\subsection{Motivation}
In classical control theory, controllers are initially designed for nominal systems without considering uncertainties or disturbances.
The problem of controller synthesis for nominal, undisturbed systems is a well-studied one. In many cases, a nominal controller features sufficient stability margins, such that it works even in the presence of uncertainties and disturbances.
However, if robustness is a more critical concern,
the controller is modified based on an uncertainty model to ensure stability and, optionally, constraint satisfaction despite worst-case disturbances.
Therefore, it is of high importance to have an understanding of disturbances and uncertainties that can arise in a given system.

Data-driven approaches are commonly used to identify at least an approximate uncertainty model.
However, as the~\emph{true} nature of uncertainties ultimately remains unknown, the notion of~\emph{identifying} an uncertainty
model is misleading.
No amount of data, no matter how large, is able to fully reflect system uncertainty. If that were the case, it would not be uncertain. For this reason, we look for ways to~\emph{unfalsify} an uncertainty model instead.

\subsection{Related Work}
The term~\emph{unfalsification} originates from~\emph{Logic of Scientific Discovery} by Austrian-British philosopher Karl Popper~\cite{Popper:1935}.
In his work, Popper calls for~\emph{falsifiability} as a necessary property for any acceptable scientific hypothesis.
Popper's work can be understood as a proposal to solve David Hume's~\emph{Problem of Induction}~\cite{Hume:1739} by postulating that scientific theories are only acceptable if they are falsifiable, \ie~if they can be disproven.
All acceptable scientific hypotheses that have not yet been falsified are not~\emph{validated}, but merely~\emph{unfalsified}.

In the context of this paper, model unfalsification refers to the process of systematically challenging mathematical models of dynamical systems based on available data to gain confidence in their validity.

\subsubsection{Model Unfalsification for Dynamical Systems Modeling}
The problem of (uncertainty) model (un-)falsification was first formalized for dynamical systems and
control by~\cite{PoollaModelValidation:1994} and~\cite{SafonovUnfalsification:1995}.
In the following, we briefly summarize their work.

\paragraph{Optimal vs Feasible Parameters}
As~\cite{KosutNonlinearUncertaintyModelUnfalsification:1997} put it, model unfalsification is basically
identical to classical system identification, but with the goal of~\emph{feasibility} instead of~\emph{optimality}.
It is important to note that system identification problems are not always set up in a way that enables the model to
recreate the given data exactly.
This distinction separates system identification from model unfalsification.
System identification strives for~\emph{optimality} in some sense, whereas unfalsification only
requires~\emph{feasibility},
but is more strict regarding the~\emph{exact} reproducibility of the data.

\paragraph{Extensions of the Idea}
The previous paragraphs make it seem like model unfalsification is always a binary problem, where a model is either falsified or unfalsified.
While this is true in principle, the idea can be extended to also include secondary conditions and objectives.
One way of doing so is to see model unfalsification as a special case of system identification that requires the model to be exact~\emph{and} in some sense
optimal~\cite{KosutUncertaintyModelUnfalsificationSysID:1995, KosutAdaptiveRobustControl:1996, KosutUnfalsificationRobustControl:2001}.

\subsubsection{Recent Contributions to the Subject}
The majority of publications on model unfalsification date from the late 1990s and early 2000s.
It was during this period that the concept, previously regarded as rather philosophical,
was formally established in control and optimization theory.
Since then, fewer publications have specifically mentioned the problem, which served as motivation for this paper. 
Nonetheless, a few noteworthy papers from recent years are to be discussed here.

A direct application of Popper's original philosophy to the example of numerical 
building modeling is given 
by~\cite{ForcelliniFalsificationValidationNumericalModels:2023} and the authors 
offer a pointed critique of the way data is currently used to inform mathematical models.
The work of~\cite{Rasul:2024} is more directly applied to optimal control and is in some ways similar to what is proposed in this paper.
They use an unfalsification approach to obtain optimal estimates for unknown parameters
of an autonomous underwater vehicle. However, they do not formulate or solve the uncertainty
model unfalsification problem as an optimization problem.
A very similar approach is taken by~\cite{Furieri:2023}, however, they solely rely on output feedback and noisy (input-output) data.
While~\cite{Furieri:2023} solve this problem via optimization, they do not consider the importance of
exact feasibility of the problem, and instead propose the use of a standard least-squares system
identification approach.
Various data-enabled predictive control techniques, such as~DeePC~\cite{Coulson:2019, Berberich:2021, Gao:2024, Xie:2024},
address this issue, but often rely on input-state data and sometimes neglect measurement noise.

\subsubsection{Semi-Infinite Programming \& Local Reduction}\label{ch:SIPLocalReductionReview}
Semi-infinite programs~(SIPs) that are also considered in this paper consist of a finite number of optimization variables and an infinite number of constraints; hence the name~\emph{semi-infinite}.
The theory of~SIP is described in more depth by~\cite{Hettich:1993}, who also provides applied examples with basic solution approaches.

A more general solution approach to SIPs is given by~\cite{Blankenship:1976}. They introduce 
the~\emph{local reduction} algorithm as a sequential approach of solving SIPs using subsets of 
constraints with finite cardinality.
This idea was extended by~\cite{ZagorowskaLocalReduction:2023} to a scenario-based approach that 
aims at finding worst-case scenarios that reduce the infinite number of constraints in the problem 
to a finite set of constraints that govern the ``worst case''.

These kinds of scenario-based approaches are not entirely new and were previously presented in the context of probabilistic solution techniques for robust
controller design~\cite{CalafioreScenarioRobustControl:2006, Grammatico:2014}. When solving~SIPs, however, new scenarios are not generated
at random but in a strategic way, such that the algorithm terminates at a local solution. This is referred to as~\emph{automatic scenario generation} by~\cite{ZagorowskaLocalReduction:2023}.

Another notable contribution to local reduction is made by~\cite{WehbehSemiInf:2024} by accounting for existence constraints and non-unique optimization variable trajectories.

\subsection{Contributions}
The main contribution of this paper is a novel optimization-based framework for data-driven uncertainty model unfalsification.
Two formulations of this problem --- \emph{optimistic} and \emph{pessimistic} --- are proposed, each providing a different perspective on disturbances and uncertainties, recognizing that uncertainty models can never fully capture true system uncertainty.
The framework also critiques conventional methods that either impose overly conservative bounds, sacrificing performance, or rely on unjustified assumptions (such as Gaussian noise).

Both the optimistic and pessimistic uncertainty model unfalsification formulations were tested
by applying an uncertain linear model formulation to nonlinear simulation data.

\section{Problem Formulation}\label{ch:ProblemFormulation}
Before the topic is discussed in detail, the most general model unfalsification problem is formulated without any further assumptions.
This formal definition is then used as the basis for several subproblems and derivations.

Assume a very general uncertain discrete-time model of a physical plant given by
\begin{subequations}\label{eq:MostGeneralModel}
    \begin{align}
        0 &= \fh(\xh_{k+1}, \xh_k, u_k, \wh_k, \thetah)\\
        \yh_{k} &= \gh(\xh_k, u_k, \wh_k, \thetah)\\
        \wh_k &\in \mathcalWh(\gammah),
    \end{align}
\end{subequations}
where~$\fh$ is a discrete-time model of the state dynamics~$\xh_k$,~$\gh$ models the measurements~$\yh_k$, $\thetah$ denotes the model parameters,~$\wh_k$ denotes the model uncertainties according to some uncertainty model~$\mathcalWh$ and $u_k$ represents the system inputs.

\begin{definition}[Model Unfalsification]\label{def:ModelUnfalsification}
    The model~\eqref{eq:MostGeneralModel} is said to be unfalsified by the~$N$-point data set~$(U, Y)\coloneq ((u_0,\dots,u_{N-1}), (y_0,\dots,y_{N-1}))$ if and only if there exists
    a parameter configuration~$\thetah\in\Thetah$ and an~$N$-point sequence of
    disturbances~$W = \{\wh_i\sim\mathcalWh(i, \gammah) \;\colon\; i = 0, 1,\dots,N-1\}$ with some~$\gammah\in\Gammah$,
    such that all prediction errors $e_k \coloneq y_k - \yh_k, k = 0,\dots,N-1$ are zero.
\end{definition}

The problem of uncertainty model unfalsification in its purest form can be formulated as the feasibility problem
\begin{subequations}\label{eq:UnfalsificationOptFeasibility}
    \begin{align}
        \text{find}\quad & \xb_0, \dots, \xb_{N-1}, \wb_0, \dots, \wb_{N-1}, \thetab, \gammab\\
        \text{\st}\quad & 0 = \fh(\xb_{i+1}, \xb_i, u_i, \wb_i, \thetab),\quad\; i = 0,\dots,N-2\\
                             & y_i = \gh(\xb_i, u_i, \wb_i, \thetab),\quad\; i = 0,\dots,N-1\\
                             & \wb_i\in\mathcalWh(\gammab),\quad\; i = 0,\dots,N-1\\
                             & \thetab\in\Thetah(U, Y),\\
                             & \gammab\in\Gammah(U, Y),
    \end{align}
\end{subequations}
which can only be solved if there exists a series of
uncertainties~$\Wb^* = (\wb^*_0, \dots, \wb^*_{N-1})$ in accordance with the
uncertainty model~$\mathcalWh$ that results in an~\emph{exact} match between
the given data~$(U, Y)$ and the model~$(\fh, \gh, \mathcalWh)$. In other words, the
model is~\emph{unfalsified} for the parameters $(\thetab^*, \gammab^*)$.

For simplicity, it is assumed that~$\mathcalWh(\gammab)$ is a time-invariant parametrized set with parameters~$\gammab$. However, more complex uncertainty models are also possible.

Then, the basic concept of uncertainty model unfalsification is developed further to allow
for optimization over the uncertainty model~$\mathcalWh(\gammab)$.
As a consequence, two general optimization problems can be
formulated: the ``optimistic'' and the ``pessimistic'' uncertainty model
unfalsification problem.

In the~\emph{optimistic} case, the optimal uncertainty model~$\mathcalWh(\gammab^*)$
is computed by minimizing some uncertainty metric
function~$r(\gammab)$. This yields the optimistic problem formulation
\begin{subequations}\label{eq:GeneralOptimisticUnfalsification}
    \begin{alignat}{3}
        \min_{\Wb, \thetab, \gammab}\quad & r(\gammab)\\
        \text{\st}\quad & \Wb \in\Omega_\thetab(\fh, \gh, U, Y)\\
                             & \wb_i\in\mathcalWh(\gammab),                    &&\quad i = 0,\dots,N-1\\
                             & \thetab\in\Thetah(U, Y),\\
                             & \gammab\in\Gammah(U, Y),
    \end{alignat}
\end{subequations}
where~$\Omega_\thetab(\fh, \gh, U, Y)$ denotes the set of permissible uncertainty trajectories given the model, parameters and data set
\begin{align}\label{eq:Omega}
    \Omega_\thetab(\fh, \gh, U, Y) &\coloneqq\nonumber\\
    \{\Wh\;\colon\; \exists \Xh&\subseteq\mathcal{X}\;\vert\; 0 = \fh(\xh_{k+1}, \xh_k, u_k, \wh_k, \thetab),\nonumber \\
    y_k = \gh(&\xh_k, u_k, \wh_k, \thetab), k = 0, \dots, N-1\}.
\end{align}

The solution of problem~\eqref{eq:GeneralOptimisticUnfalsification}
consists of the uncertainty trajectory~$\Wb^*$ that results in the optimal
uncertainty model~$\mathcalWh(\gammab^*)$ with respect to the uncertainty
metric function~$r(\gammab)$.
The optimal state trajectory~$\Xb^*$ is computed as a byproduct since it
is needed for computing the set~$\Omega_\thetab$.

Since only a single, but not necessarily unique, optimal uncertainty trajectory~$\Wb^*$ is considered in this problem formulation,
it tends to result in very optimistic uncertainty models; hence the name.

Conversely, a~\emph{pessimistic} approach can be taken for the same underlying problem.
In this case the goal is not to determine~\emph{one} uncertainty trajectory that fits the model, but to make sure that the optimal uncertainty model captures~\emph{all} possible uncertainty trajectories allowed by the given data~$(U, Y)$ and the model~$(\fh, \gh)$.
This is achieved by adding a ``for all'' constraint to the optimization problem resulting in the following semi-infinite optimization program (SIP)
\begin{subequations}\label{eq:GeneralPessimisticUnfalsification}
    \begin{align}
        \min_{\substack{\gammab\in\Gammah(U, Y),\\\thetab\in\Thetah(U, Y)}}\quad & r(\gammab)\\
        \text{\st}\quad & W \subseteq \mathcalWh(\gammab),\, \forall W\in\Omega_\thetab(\fh, \gh, U, Y),
    \end{align}
\end{subequations}
where the set~$\Omega_\thetab$ represents the set of all uncertainty trajectories that are permissible under the given model and data set for a parametric configuration~$\thetab$.

\section{Theoretical Analysis}
\subsection{Optimistic Uncertainty Model Unfalsification}
\paragraph{Non-Uniqueness of Solutions}
The optimistic uncertainty model unfalsification determines a ``minimum-volume''
set~$\mathcalWh(\gammab^*)$ with respect to the given ``volume''
function~$r(\gammab)$
enclosing all elements of~\emph{a} permissible uncertainty trajectory~$\Wb^*$.
Hence, the uncertainty realization~$\Wb^*$ and the model parameters~$(\gammab^*, \thetab^*)$ leading to the minimum-volume enclosing set~$\mathcalWh(\gammab^*)$ are not necessarily unique. In practice, non-unique solutions mostly arise in very small problems.
\paragraph{Monotonicity of the Optimal Objective Value}
The number of constraints in problem~\eqref{eq:GeneralOptimisticUnfalsification} is driven
by the number of data points~$(U, Y)$. Hence, it is worth investigating how the optimistic
uncertainty model~$\mathcalWh(\gammab^*)$ and the optimal objective
value~$r^* \coloneqq r(\gammab^*)$ behave when more data is added.
\begin{proposition}
Given a data set~$(U, Y)$ and a solution to the optimistic problem~\eqref{eq:GeneralOptimisticUnfalsification}
with optimal value $r^*$,
a larger data set $(U_+, Y_+)$ with~$U \subset U_+$
and $Y \subset Y_+$ will lead to an optimal value $r^*_+$
that is larger or equal to the original optimal
objective value, i. e. $r^* \leq r^*_+$.
\end{proposition}

\begin{proof}
Adding data points to the data set~$(U, Y)$ adds constraints to the set of permissible uncertainty trajectories~$\Omega$, and previously permissible trajectories may no longer be allowed.
The resulting set~$\Omega_+$ is smaller or equal to the original set,~\ie $\Omega_+ \subseteq \Omega$. Hence, the feasible region of problem~\eqref{eq:GeneralOptimisticUnfalsification} either remains the same,
if all previously permissible uncertainty trajectories remain permissible, or the feasible region shrinks.
As a result, the optimal objective value~$r^*$ also either remains the same or increases if the added data
rules out solutions that were previously feasible.
\end{proof}

\subsection{Pessimistic Uncertainty Model Unfalsification}
\paragraph{Necessary Conditions for Feasibility}
The ``for all'' constraint in the pessimistic uncertainty model unfalsification
problem~\eqref{eq:GeneralPessimisticUnfalsification} can only be satisfied if
the set~$\Omega$ is bounded. Hence, this is a necessary condition for feasibility.
In order to formalize the relationship between the model~$(\fh, \gh)$, the given data~$(U, Y)$ and the
bounds of the set~$\Omega$, we introduce the concept of~\emph{bounded uncertainty inferability}.

\begin{definition}\label{def:BoundedUncertaintyInferability}
If the set of all permissible uncertainty trajectories~$\Omega$ defined
in~\eqref{eq:Omega} is bounded for some data
set~$D \coloneqq (U, Y)$,~\emph{bounded uncertainty} is~\emph{inferable} from
the data $D$ for the model~$(\fh, \gh)$.
\end{definition}

Without specific knowledge about~$(\fh, \gh)$ and $(U, Y)$ it is difficult
to determine whether bounded uncertainty inferability holds.
A more concrete statement can be made for linear systems, which will be presented later on in section~\ref{ch:LinearSystems}.

From here on, we assume that bounded uncertainty inferability is satisfied and the feasibility of the pessimistic uncertainty unfalsification problem is given.

\paragraph{Strict Order Reflection}
Assume that the uncertainty metric function~$r(\gamma)$ for the uncertainty model~$\mathcalWh(\gamma)$ satisfies~\emph{monotonicity}
\begin{align}\label{eq:GrowingSets}
    \mathcalWh(\gamma_1) \subseteq \mathcalWh(\gamma_2)\;\implies\;r(\gamma_1) \leq r(\gamma_2),
\end{align}
\ie~if regions are added to a set its uncertainty metric function must also increase.
This condition is automatically satisfied if~$r(\gamma)$ is a valid set measure 
of~$\mathcalWh(\gamma)$.

If the converse is true as well, namely the uncertainty metric function~$r(\gamma)$ is strictly order-reflecting with respect
to the set~$\mathcalWh(\gamma)$,~\ie
\begin{align}
    r(\gamma_1) \leq r(\gamma_2)\;\implies\;
    \mathcalWh(\gamma_1) \subseteq \mathcalWh(\gamma_2),\,\forall\gamma_1, \gamma_2 \in\Gamma,
\end{align}
then \eqref{eq:GeneralPessimisticUnfalsification} can be written as an equivalent min-max problem
\begin{subequations}\label{eq:GrowingSetsOptProb}
    \begin{align}
        \min_{\substack{\gammab\in\Gammah(U, Y)\\\thetab\in\Thetah(U, Y)}}\max_{\Wb}\quad & r(\gammab)\\
        \text{\st}\quad    & \Wb\in\Omega_\thetab(\fh, \gh, U, Y)\\
                                & \wb_i \in\mathcalWh(\gammab),\quad i = 0,\dots,N-1,
    \end{align}
\end{subequations}
since the solution is fully determined by a single ``worst-case'' uncertainty trajectory.

This problem poses significant challenges of its own, but has a finite number of constraints and optimization variables. Therefore, no SIP solving algorithms need to be used.

\begin{proposition}\label{prop:InverseMonotonicity}
    If~$(\fh, \gh)$ and~$(U, Y)$ fulfill bounded uncertainty inferability as defined in
    Definition~\ref{def:BoundedUncertaintyInferability} and condition~\eqref{eq:GrowingSets} is met,
    problem~\eqref{eq:GrowingSetsOptProb} is equivalent to problem~\eqref{eq:GeneralPessimisticUnfalsification}.
\end{proposition}
\begin{proof}
Let~$\gamma^*_k$ be the optimal parameter for the parameterized set~$\mathcalWh(\gamma)$ enclosing the $k$-th permissible
uncertainty trajectory from the set~$\Omega_\theta$, denoted as~$\Omega_{\theta, k}$.
Hence,~$\gamma^*_k$ is a solution to
\begin{subequations}\label{eq:OptimalGammak}
    \begin{align}
        \gamma^*_k \in \argmin_{\gammab\in\Gammah(U, Y)}\quad   & r(\gammab)\\
        \text{\st}\quad                                & w_i \in\mathcalWh(\gammab), \quad i = 0,\dots,N-1\\
                                                       & W = \Omega_{\theta, k}(\fh, \gh, U, Y).
    \end{align}
\end{subequations}

Let~$\gamma^*_\textnormal{max}$ be the parameters among all~$\gamma^*_k$ that lead to the largest value of~$r(\gammab)$
\begin{subequations}\label{eq:MaximiumRGammak}
\begin{align}
    r(\gamma^*_\textnormal{max}) \geq r(\gamma^*_k),\;\forall \gamma^*_k \in \Gamma^*_\Omega, \gamma^*_\textnormal{max} \in \Gamma^*_\Omega,
\end{align}
\end{subequations}
where~$\Gamma^*_\Omega$ denotes the set of all~$\gamma^*_k$ that are solutions to problem~\eqref{eq:OptimalGammak} for a given~$\Omega_\theta(\fh, \gh, U, Y)$.

Since bounded uncertainty inferability is assumed to be satisfied,~$r(\gamma^*_{\textrm{max}})$ is bounded.
Given condition~\eqref{eq:GrowingSets}, the set~$\mathcalWh(\gamma^*_{\textrm{max}})$
is a superset of~$\mathcalWh(\gamma^*_k)$ for all~$\gamma^*_k\in\Gamma_\Omega^*$.
Hence,
\begin{align}
    W \subseteq \mathcalWh(\gamma^*_{\textrm{max}}),\;\forall W\in\Omega_\theta(\fh, \gh, U, Y),
\end{align}
in other words, all permissible uncertainty trajectories are fully contained in~$\mathcalWh(\gamma^*_{\textrm{max}})$.

Since~$r(\gamma^*_\textnormal{max})$ is already a solution of a minimization problem, there exists no~$\gamma^\prime$ with~$r(\gamma^\prime) < r(\gamma^*_\textnormal{max})$
that also satisfies~\eqref{eq:MaximiumRGammak}.
Thus, the objective and the constraints of the original~SIP~\eqref{eq:GeneralPessimisticUnfalsification} are satisfied.
Finding the maximum across solutions of a minimization problem is exactly what is done in problem~\eqref{eq:GrowingSetsOptProb}.
Hence, problems~\eqref{eq:GeneralPessimisticUnfalsification} and~\eqref{eq:GrowingSetsOptProb} are equivalent.
\end{proof}

\subsection{Duality of Optimistic and Pessimistic Problems}
\begin{proposition}
Strong duality between the optimistic and pessimistic uncertainty model unfalsification problems holds if the permissible uncertainty realization is unique, \ie~$\card{\Omega} = 1$.
\end{proposition}
\begin{proof}
Assuming that there exists exactly one permissible uncertainty realization, \ie~$\Omega = \{W_\textnormal{unique}\}$, the optimistic problem~\eqref{eq:GeneralOptimisticUnfalsification} can be rewritten as
\begin{subequations}\label{eq:ReformulatedOptProblemUniqueW}
    \begin{align}
        \min_{\gammab}\quad & r(\gammab)\\
        \text{\st}\quad & W_\textnormal{unique} \subseteq \mathcalWh(\gammab),
    \end{align}
\end{subequations}
and the pessimistic problem~\eqref{eq:GeneralPessimisticUnfalsification} becomes
\begin{subequations}
    \begin{align}
        \min_{\gammab}\quad & r(\gammab)\\
        \text{\st}\quad & W\subseteq \mathcalWh(\gammab),\; \forall W\in\{W_\textnormal{unique}\},
    \end{align}
\end{subequations}
which is identical to the reformulated optimistic problem~\eqref{eq:ReformulatedOptProblemUniqueW}.

Hence, given the uniqueness of the permissible uncertainty realization, the optimistic and pessimistic problem formulations are identical,
therefore, strong duality must hold.
\end{proof}

The more useful result is that  -- with mild assumptions -- a similar statement can be made about weak duality of the optimistic and pessimistic problems.

\begin{proposition}\label{prop:WeakDuality}
Suppose that the following holds:
\begin{enumerate}
    \item both the optimistic and pessimistic problems are feasible and have unique solutions.
    \item the set~$\mathcalWh(\gamma)$ is convex for all~$\gamma$.
    \item $\mathcalWh^*_\mathcal{X}$ enclosing a set~$\mathcal{X}$ with minimum volume $r$ satisfies monotonicity, \ie~$\mathcal{A} \subseteq \mathcal{B}\implies \mathcalWh^*_\mathcal{A} \subseteq \mathcalWh^*_\mathcal{B}$.
    \item the function~$r(\gamma)$ satisfies monotonicity, \ie~$\mathcalWh(\gamma_1) \subseteq \mathcalWh(\gamma_2) \implies r(\gamma_1) \leq r(\gamma_2)$.
\end{enumerate}
then weak duality between the optimistic and the pessimistic uncertainty model unfalsification problems holds.
\end{proposition}

\begin{proof}
Given a set of permissible uncertainty trajectories~$\Omega$, let~$(\gammab_\textnormal{opt}^*, \Wb_\textnormal{opt}^*)$ be the solution of the optimistic problem with
optimal objective value~$r(\gammab_\textnormal{opt}^*)$, whereas the solution of the pessimistic problem is assumed to be given by~$\gammab_\textnormal{pess}^*$ with optimal
objective value~$r(\gammab_\textnormal{pess}^*)$.
From the optimistic problem follows that~$\Wb_\textnormal{opt}^* \subseteq \mathcalWh(\gammab_\textnormal{opt}^*)$ and
since~$\Wb_\textnormal{opt}^* \in\Omega$, $\Wb_\textnormal{opt}^* \subseteq \mathcalWh(\gammab_\textnormal{pess}^*)$ must also be satisfied.

By construction of the pessimistic problem, the set~$\mathcalWh(\gammab_\textnormal{pess}^*)$ must contain~\emph{at least} the trajectory~$\Wb_\textnormal{opt}^*$.
In the special case where~$\Omega$ solely consists of a unique uncertainty trajectory, strong duality holds, as proven above.
In this case, the unique uncertainty trajectory is equivalent to~$\Wb_\textnormal{opt}^*$. If~$\card{\Omega} > 1$, the set~$\mathcalWh(\gammab_\textnormal{pess}^*)$ remains a superset of~$\Wb_\textnormal{opt}^*$, but also contains all other trajectories given in~$\Omega$.
This implies that~$\mathcalWh(\gammab_\textnormal{pess}^*)$ must contain at least one point that is not contained in~$\mathcalWh(\gammab_\textnormal{opt}^*)$, assuming that trajectories in~$\Omega$ are mutually distinct.

However, given monotonicity of the set~$\mathcalWh^*_\mathcal{X}$, the pessimistic enclosing set $\mathcalWh(\gammab^*_\textnormal{pess})$ must fully contain
the optimistic enclosing set $\mathcalWh(\gammab_\textnormal{opt}^*)$
\begin{align}
    \mathcalWh(\gammab_\textnormal{opt}^*) \subseteq \mathcalWh(\gammab_\textnormal{pess}^*).
\end{align}

Then, by monotonicity of~$r$, the optimal objective value of the pessimistic problem upper bounds any solution of the optimistic problem
\begin{align}
    \mathcalWh(\gammab_\textnormal{opt}^*) \subseteq \mathcalWh(\gammab_\textnormal{pess}^*) \implies r(\gammab_\textnormal{opt}^*) \leq r(\gammab_\textnormal{pess}^*),
\end{align}
and therefore the two problems satisfy weak duality.
\end{proof}

\subsection{Linear Systems}\label{ch:LinearSystems}

The previously introduced concept of bounded uncertainty inferability shares some similarities to observability in classical control theory.
Therefore the~\emph{uncertainty observability matrix} is introduced analogously to the
observability matrix.

\begin{definition}[Uncertainty Observability Matrix]\label{def:UncertaintyObservabilityMatrix}
Given the linear system
\begin{subequations}\label{eq:LinearModelLinearUncertainties}
    \begin{align}
        x_{k+1} &= Ax_k + Bu_k + Ew_k\\
        y_k &= C x_k + D u_k + F \vh_k,
    \end{align}
\end{subequations}
the measurements~$Y$ can be related to the system inputs by the system of linear equations
\begin{align}\label{eq:UncertaintyObservabilityEquations}
    y_k &- C\sum_{i=0}^{k-1}A^i B u_{k-i-1} - D u_k - F \vh_k\nonumber\\
    &= C A^k \xh_0 + C \sum_{i=0}^{k-1} A^i E \wh_{k-i-1},\,\, k = 0,\dots,N-1.
\end{align}

Based on this system of equations, the~\emph{uncertainty observability matrix}~$\mathcal{E}$ is defined as
\begingroup
\footnotesize
\begin{align}
    \mathcal{E} &\coloneq \begin{bmatrix}
        C        & 0         & 0         & \cdots & 0\\
        CA       & CE        & 0         & \cdots & 0\\
        CA^2     & CAE       & CE        & \cdots & 0\\
        \vdots   & \vdots    & \vdots    & \ddots &\vdots\\
        CA^{N-1} & CA^{N-2}E & CA^{N-3}E & \cdots & CE
    \end{bmatrix}.
\end{align}
\endgroup
\end{definition}
\begin{remark}
Input feedthrough~$D u_k$ has no effect on~$\mathcal{E}$.
\end{remark}
\begin{fact}\label{fact:Fact1}
For~$\vh_k \equiv 0$, the system of
equations~\eqref{eq:UncertaintyObservabilityEquations} has a unique solution if
and only if~$\mathcal{E}$ is square and nonsingular.
\end{fact}
\begin{fact}\label{fact:Fact2}
If~$\mathcal{E}$ is nonsingular and all measurement noise is zero~($\vh_k \equiv 0)$, the set of permissible uncertainty trajectories~$\Omega$ contains exactly one uncertainty trajectory.
\end{fact}

With~$A\in\R{n\times n}$,~$C\in\R{p\times n}$ and~$E\in\R{n\times d_w}$, the dimensions of~$\mathcal{E}$ are~$p (N+1) \times (n + d_w N)$.
Hence, for specific cases, it is possible for~$\mathcal{E}$ to be square and nonsingular.
In a lot of practical cases, however,~$\mathcal{E}$ is reducible and results in a slightly shorter matrix after reduction.
Then, the system of equations~\eqref{eq:UncertaintyObservabilityEquations} is under-determined and therefore has unbounded solutions. In this case, additional assumptions are necessary in order to ensure bounded uncertainty inferability.
\begin{fact}\label{fact:Fact3}
If~$\vh_k\equiv 0$ and~$\mathcal{E}$ are irreducible and short (fewer rows than columns) then~$\Omega$ is unbounded, and bounded uncertainty inferability is not satisfied.
\end{fact}

Facts~\ref{fact:Fact1} through~\ref{fact:Fact3} neglect the effect of
measurement noise. Under the assumption that measurement noise~$\vh_k$ affects
the system in a linear way (as in~\eqref{eq:LinearModelLinearUncertainties}), a
slightly weaker statement about the boundedness of the set of possible uncertainty trajectories~$\Omega$ can be made.
\begin{proposition}
For linear models of the form~\eqref{eq:LinearModelLinearUncertainties}
with~$\det{\mathcal{E}} \neq 0$ and~$\vh_k \in \mathcalVh$, where~$\mathcalVh$ is
a bounded set, the set of all permissible uncertainty trajectories~$\Omega$
is bounded.
\end{proposition}

\begin{proof}
Take the system of equations~\eqref{eq:UncertaintyObservabilityEquations}
with~$\vh_k \neq 0$.
Then, since~$\det{\mathcal{E}} \neq 0$, there exists a unique solution to
the system of equations for~\emph{any} given~$\vh_k\in\mathcalVh$.
Namely,
\begin{align}\label{eq:ExplicitWSolution}
    \begin{bmatrix}
        \xh_0\\
        \wh_0\\
        \vdots\\
        \wh_{N-2}
    \end{bmatrix} = &\underbrace{\mathcal{E}^{-1} \begin{bmatrix}
        y_0 - Du_0\\
        y_1 - CBu_0 - Du_1\\
        \vdots\\
        y_{N-1} - C\sum_{i=0}^{N-2}A^i B u_{k-i-1} - D u_{N-1}
    \end{bmatrix}}_{\textrm{``homogeneous'' solution for $\vh_k \equiv 0$}}\nonumber\\
    &\quad\quad -\mathcal{E}^{-1} \begin{bmatrix}
        F\vh_0\\
        F\vh_1\\
        \vdots\\
        F\vh_{N-1}
    \end{bmatrix},\quad \vh_k\in\mathcalVh.
\end{align}

Therefore, the resulting trajectory~$\{\wh_k\}_{k = 0,\dots,N-2}$ is a linear combination of elements in the data set~$(U, Y)$ and elements
from the measurement noise trajectory~$\{\vh_k\}_{k=0,\dots,N-1}$.

Take a linear combination of two vectors~$a$ and~$b$
\begin{align}
    c = \alpha a + \beta b,
\end{align}
where~$\alpha, \beta\in\justR$,~$a\in \mathcal{A}$ and~$b\in\mathcal{B}$.
Making use of the triangle inequality, the norm of~$c$ can be computed as
\begin{subequations}
\begin{align}
    \lVert c\rVert = \lVert \alpha a + \beta b\rVert &\leq \lVert\alpha a\rVert + \lVert \beta b\rVert\\
                                                     &\leq \lvert \alpha\rvert\cdot \lVert a\rVert + \lvert \beta \rvert \cdot \lVert b\rVert.
\end{align}
\end{subequations}

Assuming that the sets~$\mathcal{A}$ and~$\mathcal{B}$ are bounded with bounds~$M_\mathcal{A}$ and~$M_\mathcal{B}$, respectively,
a bound for any linear combination~$c$ can be formulated as
\begin{subequations}
    \begin{align}
                    (\lVert a\rVert \leq M_\mathcal{A},\;\forall a\in\mathcal{A})
        \;\wedge\;    (\lVert b\rVert \leq M_\mathcal{B},\;\forall b\in\mathcal{B})\\
        \Rightarrow\; \lVert c\rVert \leq \lvert\alpha\rvert\cdot M_\mathcal{A} + \lvert\beta\rvert\cdot M_\mathcal{B} < \infty.
    \end{align}
\end{subequations}

Going back to the original problem, it follows that if the data set~$(U, Y)$ and the measurement noise set~$\mathcalVh$ are bounded, the set~$\Omega$ containing
all possible uncertainty trajectories given by~\eqref{eq:ExplicitWSolution} must also be bounded.
\end{proof}

\section{Numerical Results}
For testing the optimistic and pessimistic uncertainty model unfalsification approaches,
a simple ``pendulum on a cart'' system was used to generate data.
As depicted in Figure~\ref{fig:PendulumOnCart}, this system consists of a one-dimensional cart carrying a nonlinear
pendulum. The equations of motion for a nearly identical system can be found in~\cite{Owen:2014}.
In this example, it is assumed that only the linear force input~$F(t)$ is available.
In order to generate meaningful data that covers a wide range of force inputs,
a persistently exciting input sequence was applied to the system and the state evolution was simulated using a self-developed Julia simulation framework~\cite{SimpleSim.jl}.

Then, the uncertain linear model~\eqref{eq:LinearModelLinearUncertainties} was used for optimistic and
pessimistic uncertainty model unfalsification. Since the Jacobians of the original nonlinear systems
were computed to set up the linear model, the model uncertainty is equivalent to the linearization error.
For this example, we chose the following simple $\ell_2$-norm ball to model the uncertainty:
\begin{align}\label{eq:L2normball}
    \mathcalWh(\gamma) = \{w\;\colon\; \lVert w\rVert_2 \leq \gamma\}.
\end{align}
An exemplary uncertainty trajectory and the fitted $\ell_2$-norm ball are visualized in Figure~\ref{fig:L2Ball}.

\begin{figure}[t]
    \vspace{-8mm}
    \centering
    \resizebox {0.6\columnwidth} {!} {
    \begin{tikzpicture}
        \draw[draw=none] (-3.5,-3.5) rectangle (3.5, 2);

        \draw[|-Latex] (0, 0.8) -- (1.5, 0.8) node[midway, above] {$x$};

        \draw[thick] (-1.0,-0.5) rectangle (1.0,0.5);  
        \draw[thick, black, fill=white] (-0.6,-0.5) circle (0.2); 
        \draw[thick, black, fill=white] (0.6,-0.5) circle (0.2);  

        \draw[fill, pattern=north east lines, draw=none] (-2.0,-1.0) rectangle (2.0,-0.7);
        \draw[thick] (-2.0,-0.7) -- (2.0,-0.7);

        \draw[thick] (0,0) -- ({3*tan(25)},-3) node[pos=0.65, above right] {$L$};

        \draw[fill=black] (0,0) circle (0.07);     
        \draw[fill=black] ({3*tan(25)},-3) circle (0.25); 
        \node at ({3*tan(25)+0.5},-3) {$m$};

        \draw[dashed] (0,0) -- (0,-3);
        \draw[-Latex] ([shift=(270:2.5)]0,0) arc (270:295:2.5) node[midway, below] {$\theta$};

        \draw[-{Latex[length=5pt]}] (-2.2, 0.0) -- (-1.0, 0.0) node[midway, above] {$F(t)$};

        \draw[-Latex] ([shift=(255:1.5)]0,0) arc (260:313:1.5) node[near start, below left] {$T(t)$};

        \draw[-{Latex[length=5pt]}] (-1.5, -2.2) -- (-1.5, -3) node[pos=0.4, right] {$g$};
    \end{tikzpicture}
    }
    \vspace{-8pt}
    \caption{Pendulum on a Cart}
    \label{fig:PendulumOnCart}
\end{figure}
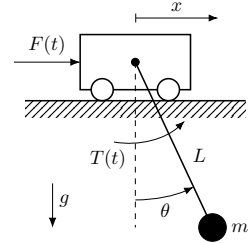

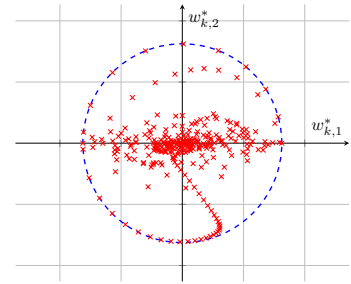
\begin{figure}[b]
    \centering
    \resizebox {0.52\columnwidth} {!} {
    \begin{tikzpicture}
        \begin{axis}[
            grid=both,
            axis lines = middle,
            axis equal,
            legend pos=north east,
            scatter/classes={%
                a={mark=x,red}%
            },
            legend style={
                legend cell align=left
            },
            xlabel=$w^*_{k,1}$,
            ylabel=$w^*_{k,2}$,
            enlargelimits = 0.2,
            xticklabel=\empty,
            yticklabel=\empty,
            scaled ticks = false,
        ]
            \pgfplotstableread[col sep=comma]{figure_data/set_L2.csv}\radiusdata
            \pgfplotstablegetelem{0}{radius}\of\radiusdata
            \let\ballradius\pgfplotsretval

            \addplot[
                scatter,
                only marks,
                scatter src=explicit symbolic,
            ] table[
                col sep=comma,
                meta=label
            ] {figure_data/data_L2.csv};

            \draw[thick, dashed, blue] (0,0) circle [radius=\ballradius];
        \end{axis}
    \end{tikzpicture}
    }
    \vspace{-5pt}
    \caption{Uncertainty trajectory in optimized $\ell_2$-ball}
    \label{fig:L2Ball}
\end{figure}

Using the data generated by the simulated nonlinear pendulum on a cart and the uncertain linear system,
both the optimistic and pessimistic uncertainty model unfalsification problems were solved for the
uncertainty model~\eqref{eq:L2normball}. 
The optimistic uncertainty model unfalsification problem can be solved directly using widely available solvers. In this case, the solver IPOPT~\cite{Ipopt} was used.
The pessimistic uncertainty model unfalsification problem~\eqref{eq:GeneralPessimisticUnfalsification} is a~\emph{semi-infinite optimization program} (SIP). The SIP was first cast into standard form and then solved using the~\emph{local reduction algorithm},
where the infinite number of constraints is replaced by a finite subset of constraints
that lead to the same optimum~\cite{Blankenship:1976, Hettich:1993, ZagorowskaLocalReduction:2023}.

Based on this numerical example, it can be confirmed that the optimal objective value of the optimistic problem grows
monotonically as more data is added to the data set.
This is not true for the pessimistic problem for which the optimal objective value can both rise and fall when adding new data.
However, it can be confirmed that the solution to the pessimistic problem upper-bounds the optimistic problem, which fits the duality of the optimistic and pessimistic problems. These results are visualized in Figure~\ref{fig:ComparisonR}.

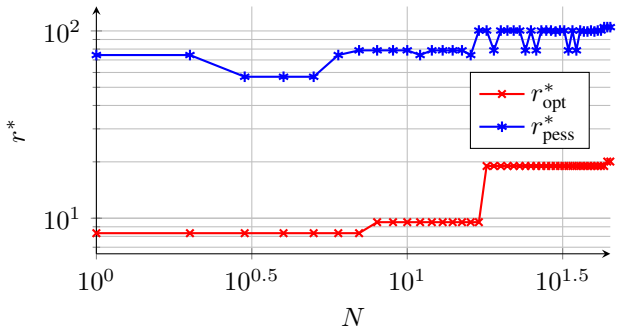
\begin{figure}[t]
    \centering
    \begin{tikzpicture}
        \begin{axis}[
            grid=both,
            xmode=log,
            ymode=log,
            legend style={
                legend cell align=left,
                at={(0.96,0.58)},
                anchor=east,
            },
            scatter/classes={%
                a={mark=x,red},%
                b={mark=asterisk,blue}%
            },
            xlabel=$N$,
            ylabel=$r^*$,
            axis x line=bottom,
            axis y line=left,
            xmax=45,
            width=0.95\columnwidth,
            height=0.55\columnwidth,
            enlarge y limits=0.1,
        ]
            \addplot[
                scatter,
                red,
                thick,
                scatter src=explicit symbolic,
            ] table[
                x=N,
                y=r_opt,
                col sep=comma,
                meta=label,
            ] {figure_data/r_and_d.csv};
            \addlegendentry{ $r_\textnormal{opt}^*$}

            \addplot[
                scatter,
                blue,
                thick,
                scatter src=explicit symbolic,
            ] table[
                x=N,
                y=r_pessi,
                col sep=comma,
                meta=label2,
            ] {figure_data/r_and_d.csv};
            \addlegendentry{ $r_\textnormal{pess}^*$}
        \end{axis}
    \end{tikzpicture}
    \vspace{-10pt}
    \caption{Optimal objective values of optimistic and pessimistic problems for nonlinear system and $\ell_2$-norm ball uncertainty set with increasing data size}
    \label{fig:ComparisonR}
\end{figure}

\section{Conclusions}
In this paper, we proposed a novel framework for data-driven uncertainty model unfalsification to allow for
better justification of assumptions regarding system uncertainties.


We present two possible techniques of choosing an unfalsified uncertainty model: an optimistic and a pessimistic one.
The proposed optimistic uncertainty model unfalsification approach determines the least conservative uncertainty model
that explains the available data.
In contrast, the pessimistic approach consists of computing the set of all permissible uncertainty trajectories given
the data, and determines the uncertainty model that represents all of those uncertainty trajectories. It was
shown that this problem is a semi-infinite program and can be solved using the local reduction algorithm. 

Additionally, theoretical  analysis revealed that the pessimistic uncertainty model unfalsification problem
is feasible under the assumption of bounded uncertainty inferability and bounded measurement noise.
Finally, it was shown that the optimistic uncertainty model represents a~\emph{strict} lower bound, thereby falsifying all
uncertainty models that are less conservative, while the pessimistic uncertainty model can be taken as a loose upper bound.
This duality between the two problems was investigated in more detail.

\section*{Acknowledgments}
The authors would like to thank Dr Andrea Carron and Prof.\ Dr Melanie Zeilinger of the Institute for Dynamic Systems and Control at ETH Z\"urich for enabling this project.

\vfill
\bibliographystyle{IEEEtran}
\bibliography{IEEEabrv,./references.bib}

\end{document}